\documentclass[11pt]{article}
\usepackage[T1]{fontenc}

\usepackage{amsmath}
\usepackage{amsthm}
\usepackage{anyfontsize}
\usepackage{authblk}
\usepackage{bbm}
\usepackage{cite}
\usepackage{enumitem}
\usepackage[para]{footmisc}
\usepackage[margin=1in]{geometry}
\usepackage{graphicx}
\usepackage{hyperref}
\usepackage{mathtools}
\usepackage{physics}
\usepackage{thm-restate}
\usepackage{titling}
\usepackage{xparse}

\usepackage{pgfplots}
\pgfplotsset{compat=1.18}

\usepackage{tikz}
\usetikzlibrary{positioning}

\usepackage{xcolor}
\definecolor{linkcolor}{RGB}{0,15,175}
\definecolor{citecolor}{RGB}{150,0,0}
\definecolor{urlcolor}{RGB}{0,15,125}
\hypersetup{colorlinks,
	linkcolor=linkcolor,
	citecolor=citecolor,
	urlcolor=urlcolor
}

\usepackage{newpxtext}
\usepackage{newpxmath}
\usepackage{libertineMono}

\newcommand{\I}{\mathbbm{1}}

\newcommand{\N}{\mathbb{N}}
\newcommand{\bit}{\{0,1\}}
\newcommand{\cartprod}{\mathbin{\square}}
\newcommand{\gammacomm}{Gram matrix relation}
\newcommand{\ketbraij}{\ketbra*{\mspace{2mu}i}{j\mspace{1mu}}_{X'}}
\newcommand{\MAJ}{\mathsf{MAJ}}
\newcommand{\AND}{\mathsf{AND}}
\newcommand{\XOR}{\mathsf{XOR}}

\let\P\relax
\let\norm\relax
\DeclareMathOperator{\P}{P}

\DeclareMathOperator{\Binom}{Binomial}
\DeclareMathOperator{\greedy}{greedy}
\DeclareMathOperator{\glo}{global}
\DeclarePairedDelimiter{\norm}{\lVert}{\rVert}

\usepackage[capitalize]{cleveref}
\theoremstyle{plain}
\newtheorem{theorem}{Theorem}
\newtheorem{lemma}[theorem]{Lemma}
\theoremstyle{definition}
\newtheorem{definition}[theorem]{Definition}

\makeatletter
\long\def\@makefntext#1{\leavevmode\@makefnmark\nobreak#1}
\makeatother

\NewDocumentCommand{\multiemail}{ >{\SplitList{,}}m m }{%
    \newif\iffirstitem%
    \firstitemtrue%
    \NewDocumentCommand{\func}{ m }{%
        \iffirstitem%
            \firstitemfalse%
        \else
            ,%
        \fi%
        \href{mailto:##1@#2}{\color{black}##1}%
    }%
    \texttt{\{%
        \ProcessList{#1}{\func}%
        \firstitemtrue%
    \}@#2}}
\NewDocumentCommand{\email}{ m }{\texttt{\href{mailto:#1}{\color{black}#1}}}

\begin{document}

\setlength{\droptitle}{-1.6cm}
\title{Local strategies are pretty good at computing Boolean properties of quantum sequences}

\newcommand{\aff}[1]{$^\text{\ref{affil#1}}$}
\author{
    Tathagata~Gupta\aff{1}, \;
    Ankith~Mohan\aff{2}, \;
    Shayeef~Murshid\aff{3},\authorcr
    Vincent~Russo\aff{4}, \;
    Jamie~Sikora\aff{2}, \;
    Alice~Zheng\aff{2}
}

\makeatletter
\hypersetup{
    pdftitle={\@title},
    pdfauthor={
        Tathagata Gupta,
        Ankith Mohan,
        Shayeef Murshid,
        Vincent Russo,
        Jamie Sikora,
        Alice Zheng
    }
}
\makeatother

\newcommand{\afftarget}[2]{\footnotetext[#1]{\phantomsection\label{affil#1}#2}}
\afftarget{1}{
    Department of Physics, Indian Institute of Technology,
    Madras, Chennai 600036, India.
    \email{tathagatagupta@gmail.com}}
\afftarget{2}{
    Department of Computer Science, Virginia Tech,
    Blacksburg, VA 24061, USA.
    \multiemail{ankithmo,sikora,alicezheng}{vt.edu}}
\afftarget{3}{
    Electronics and Communication Sciences Unit, Indian Statistical Institute,
    Kolkata 700108, India.
    \email{shayeef\_r@isical.ac.in}}
\afftarget{4}{
    Unitary Foundation,
    San Francisco, California 94104, USA.
    \email{vincentrusso1@gmail.com}}

\date{March~6, 2026}
\maketitle

\begin{abstract}
    Quantum memory is a scarce and costly resource, yet little is known about which learning tasks remain feasible under severe memory constraints.
    We study the problem of computing global properties of quantum sequences when quantum systems must be measured individually, without storing or jointly processing them.
    In our setting, a bit string $x \in \bit^n$ is encoded into an $n$-qubit product state
    $\ket{\psi_{x_1}} \otimes \cdots \otimes \ket{\psi_{x_n}}$, and the goal is to infer $f(x) \in \bit$ from measurements of this quantum encoding.
    We consider a simple local strategy, which we call the \emph{greedy strategy}, that applies the same optimal single-system measurement independently to each subsystem and then infers $f(x)$ from the outcomes.
    Our main result gives a complete characterization of when the greedy strategy is optimal: it achieves the same maximum success probability as an unrestricted global measurement if and only if the target Boolean function is affine (in all but finitely many cases).
    We establish a universal performance guarantee for general Boolean functions, showing that the success probability of the greedy strategy is always at least the square of the optimal global success probability, in direct analogy with the Barnum-Knill bound for the pretty good measurement.
    These results demonstrate that even under extreme memory constraints, simple local measurement strategies can remain provably competitive for learning global properties of quantum sequences.
\end{abstract}

\section{Introduction}
Suppose that a learner knows two quantum states $\ket{\psi}$ and $\ket{\phi}$, and receives a sequence of the form
\begin{equation}
    \ket{\psi} \otimes \ket{\phi} \otimes \ket{\psi} \otimes \cdots \otimes \ket{\psi}.
\end{equation}
While each position is known to be either $\ket{\psi}$ or $\ket{\phi}$, the identity of each state is unknown.
Can one determine whether the number of $\ket{\psi}$ states exceeds the number of $\ket{\phi}$ states, or whether the total number of $\ket{\psi}$ states is even?
These examples highlight a basic task: learning coarse-grained properties of a quantum sequence without necessarily determining the sequence itself.

We study this problem in two contrasting resource regimes.
In the \emph{full-memory} setting the entire (finite) sequence may be stored so that an optimal joint measurement can be performed upon receiving all of the states.
In the \emph{memoryless} setting, neither quantum nor classical memory is available and each state must be measured immediately, using only predetermined local measurements.
This work investigates which binary properties (properties which can be either true or false) can still be learned optimally under such stringent memory constraints.

Quantum state discrimination (see~\cite{chefles1998quantum, barnett2009quantum,bae2015quantum,bergou2004discrimination} for reviews) provides the broader context for this problem.
Here, a learner is presented with a quantum state drawn from a known ensemble, i.e., a collection of quantum states together with their \emph{a priori} probabilities, and the goal is to infer the identity of a given state.
This is a fundamental task in quantum communication and cryptography, and the impossibility of perfect discrimination among non-orthogonal quantum states makes this a highly nontrivial problem \cite{nielsen2010quantum}.
This has fueled nearly five decades of work \cite{helstrom1969quantum}, producing extensive catalogs of state ensembles and their optimal discrimination probabilities \cite{sun2001optimum,sugimoto2010complete,herzog2004minimum,herzog2005optimum,herzog2015optimal}.

Two paradigms have received particular attention: minimum-error discrimination \cite{bae2013structure}, where the learner outputs a guess with the goal of minimizing the overall probability of error, and unambiguous discrimination \cite{ivanovic1987differentiate, dieks1988overlap, peres1988differentiate}, which forbids incorrect guesses altogether at the cost of allowing inconclusive outcomes.
Despite decades of study, closed-form solutions are known only for some particular ensembles in both paradigms \cite{jaeger1995optimal,barnett2009quantum}.
Nevertheless, the optimal discrimination probability for an arbitrary ensemble can be efficiently obtained via a semidefinite program \cite{eldar2003designing,eldar2003semidefinite}.

Many information-processing and quantum-computing protocols involve discriminating sequences of quantum states.
The well-known BB84 quantum key distribution scheme is an example where Alice sends Bob qubits chosen at random from the set $\{\ket{0}, \ket{1}, \ket{+}, \ket{-}\}$ \cite{bennett2014quantum}.
From Bob's perspective, the received signal is a quantum sequence whose individual components are drawn from this set.
The central question in such scenarios concerns the resources required to optimally discriminate the given sequence, which might vary significantly across different sequence types.
For example, when discriminating two pure states given a finite number of copies, adaptive local strategies that make limited use of classical memory to condition later measurements based on earlier outcomes achieve the same performance as optimal joint measurements, whereas memoryless strategies generally do so only asymptotically \cite{acin2005multiple}.
However, for two mixed states, joint measurements can strictly outperform all local strategies even in the asymptotic limit \cite{calsamiglia2010local}.
In the quantum change-point detection problem, joint measurements provide an advantage in the minimum-error paradigm \cite{sentis2016quantum}, while for unambiguous detection the optimality of local versus global strategies depends on the overlap between the two states \cite{sentis2017exact}.
For identification of symmetric, pure states, online strategies have been proposed that often achieve optimality using classical memory to store only the last obtained measurement outcome \cite{sentis2022online}.

More recently, for problems where the goal is to identify the entire sequence composed of independently prepared states, it has been shown that memoryless strategies are optimal in the sense that they achieve the same performance as joint measurements, for both minimum-error and unambiguous discrimination \cite{gupta2024unambiguous,gupta2024optimal}.
This observation naturally raises the question of whether such equivalence persists when the discrimination task seeks to learn only specific properties or coarse-grained features of the sequence, rather than its exact identity.

Motivated by this, we initiate a study of learning Boolean properties of quantum sequences.
In a classical information-processing/computation setup, data are represented by bit strings, and computing Boolean functions of these strings, such as the $\XOR$, is elementary; the main challenges lie instead in algorithmic efficiency and robustness of protocols.
In the quantum setting, by contrast, computing even a simple function becomes non-trivial when no classical description of the data is available.
We consider which properties can be optimally learned using memoryless measurement strategies, and which fundamentally require joint measurements across the sequence.
The paradigm we focus on is minimum-error, where the learner might obtain an incorrect value for the property of the sequence and their goal is to minimize this error, and we completely characterize the Boolean functions for which fixed local measurements are optimal.

\section{Problem formulation and main results}\label{sec:main}
In this section we formulate our problem as a task of quantum state discrimination and present the main results.
The section ends with a short guide to the structure of the paper.

\subsection{Problem formulation}

In classical computing, data are modeled by bit strings.
In our setting, we consider encoding the classical symbols 0 and 1 into two arbitrary, but fixed and known quantum states $\ket{\psi_0}$ and $\ket{\psi_1}$.
This induces a map from an $n$-bit string $x \in \bit^n$ to a product quantum state of length $n$,
\begin{equation}
    x = x_1 x_2 \ldots x_n \in \bit^n \longmapsto  \ket{\psi_x} \coloneqq \ket{\psi_{x_1}} \otimes \ket{\psi_{x_2}} \otimes \cdots \otimes \ket{\psi_{x_n}}.
\end{equation}
Any Boolean property of the resulting quantum sequence can be described by a Boolean function $f:\bit^n \to \bit$.
Learning this property, equivalently, determining the value of $f(x)$ does not require full identification of the underlying string $x$.
Rather, it suffices to discriminate between the two sets of quantum sequences corresponding to inputs for which $f(x)=0$ and for which $f(x)=1$.
Thus, given an unknown quantum sequence $\ket{\psi_{x}}$, the task is to decide whether the classical label $x$ belongs to $S_0$ or $S_1$, where $S_0 = f^{-1}(0)$ and $S_1 = f^{-1}(1)$ are the pre-images of 0 and 1 respectively.
Without loss of generality we can consider $\ket{\psi_0}$ and $\ket{\psi_1}$ to be qubit states with a nonnegative inner product \cite[Section~3.1]{barnett2009quantum}.
Furthermore, since orthogonal and identical states represent trivial discrimination cases, we restrict the inner product to the open interval $\braket{\psi_0}{\psi_1} = s \in (0,1)$.

To learn the property\footnote{ Henceforth, we use the terms ``property'' and ``function'' interchangeably.
} $f$, the learner must discriminate between two mixed quantum states $\sigma_0$ and $\sigma_1$ where $\sigma_i$ is the density matrix obtained by mixing the states $\left\{\ketbra{\psi_x}{\psi_x} : f(x) = i\right\}$.
Assuming a uniform prior over the input strings in $\bit^n$, these mixed states can be written as
\begin{equation}\label{eq:states}
    \sigma_0 = \frac{1}{\abs{S_0}} \sum_{x \in S_0} \ketbra{\psi_x}{\psi_x} \quad \text{and} \quad \sigma_1 = \frac{1}{\abs{S_1}} \sum_{x \in S_1} \ketbra{\psi_x}{\psi_x},
\end{equation}
which occur with prior probabilities $p_0 = \abs{S_0}/2^n$ and $p_1 = \abs{S_1}/2^n$.
We refer to the resulting ensemble
\begin{equation}\label{eq:boolean-ensemble}
    \mathcal{E} = \{(p_0,\sigma_0),(p_1,\sigma_1)\}
\end{equation}
as the Boolean ensemble.
In this work, we consider the \emph{minimum-error} evaluation of the function $f$, where the learner measures the unknown quantum state $\ket{\psi_x}$ to infer the value of $f(x)$, seeking to minimize the probability of error.
Operationally, this corresponds to discriminating between $\sigma_0$ and $\sigma_1$ using a two-outcome POVM $\mathcal{M} = \{ M_0, M_1 \}$, where the outcome $i$ is interpreted as the guess $f(x)=i$.
Conditioned on the state being $\sigma_i$, the probability of a correct guess is $\Tr(\sigma_i M_i)$.
The optimal success probability is obtained by maximizing the average probability of a correct decision over all such POVMs:
\begin{equation}\label{SDP:qsd}
    \max \left\{
        p_0 \Tr(\sigma_0 M_0) + p_1 \Tr(\sigma_1 M_1) : M_0 + M_1 = \I, M_0, M_1 \geq 0
    \right\}.
\end{equation}
This is a semidefinite program and generalizes to discrimination among any finite ensemble of quantum states.
Since we are only dealing with two states here, the optimal minimum-error success probability is given by the Helstrom bound \cite{helstrom1969quantum,holevo1976investigations}
\begin{equation}\label{eq:tracenorm}
    \frac{1}{2} + \frac{1}{2} \norm*{ p_0 \sigma_0 - p_1 \sigma_1 }_{1},
\end{equation}
where $\norm{X}_{1} = \Tr\bigl(\sqrt{X^*X}\bigr)$ is the trace norm.

An optimal POVM achieving the maximum success probability may require a joint measurement acting on the entire $n$-qubit state $\ket{\psi_x}$, which is technologically challenging \cite{lvovsky2009optical,ma2021one,conlon2023approaching}.
A more feasible alternative is a local measurement strategy, in which each state is measured individually.
Such a strategy effectively extracts classical information from each subsystem and uses it to infer the value of $f(x)$.
In this work, we focus on a simple procedure: a fixed local strategy that applies a predetermined measurement to each subsystem, independent of all of the previous outcomes, thus requiring neither classical nor quantum memory.
A natural local measurement to consider is the one that aims to optimally distinguish between $\ket{\psi_0}$ and $\ket{\psi_1}$ at each position.
Given any Boolean function $f$, this strategy learns each bit optimally, constructs a guess for the encoded string $y$, and evaluates $f(y)$.
This can also be seen as a \emph{non-adaptive greedy} strategy\footnote{
    Note, however, that this notion of greediness may not be optimal in the class of local memoryless strategies for computing the overall function value.
}.
One might be tempted to think that this na{\"\i}ve and simple strategy would perform rather poorly when compared with a global measurement.
Our first result shows that, remarkably, this strategy performs rather well---but first we establish some notation.
\begin{definition}
    Given a positive integer $n$, a function $f : \bit^n \to \bit$, and two encoding qubit states $\ket{\psi_0}$ and $\ket{\psi_1}$ with inner product $s = \braket{\psi_0}{\psi_1}$, the success probability of learning $f$ using a strategy ``strat'' when each $x \in \bit^n$ is chosen uniformly at random is denoted by $\P(n,f,s,\mathrm{strat})$.
\end{definition}
Throughout the paper, we denote the fixed local (greedy) strategy and a globally optimal strategy by ``$\greedy$'' and ``$\glo$'', respectively.

\subsection{Main results}

We now state our first main result which establishes a universal performance guarantee for the greedy strategy.
\begin{theorem}\label{thm:PGM}
    For all Boolean functions $f$, positive integers $n$, and $s \in (0,1)$, the performance of the greedy strategy can be bounded below by the following relation
    \begin{equation}
        \P(n,f,s,\glo)^2 \leq \P(n,f,s,\greedy).
    \end{equation}
\end{theorem}

The proof consists in showing that the greedy strategy is implemented by the pretty good measurement (PGM) \cite{hausladen1994pretty} for discriminating the
Boolean ensemble $\mathcal{E} = \{(p_0, \sigma_0), (p_1, \sigma_1)\}$.
The bound then follows from Barnum and Knill's theorem~\cite{barnum2002reversing,watrous2018theory}.
This is presented in \cref{subsec:greedy=pgm}.

This raises a natural question: for which Boolean functions does the greedy strategy actually achieve the global optimum?
The next result identifies a class of functions that satisfy this.
\begin{restatable}[Sufficiency]{theorem}{sufficient}\label{thm:sufficient}
    Let $f : \bit^n \to \bit$ be of the form
    \begin{equation}\label{eq:affine}
        f(x_1,\ldots,x_n) = b_0 \oplus b_1 x_1 \oplus b_2 x_2 \oplus \cdots \oplus b_n x_n,
    \end{equation}
    for some $b_0, b_1, \ldots, b_n \in \bit$.
    Then, the greedy strategy is globally optimal, that is
    \begin{equation}
        \P(n,f,s,\greedy)=\P(n,f,s,\glo)
    \end{equation}
    for all $s \in (0,1)$.
\end{restatable}

The functions of the form in Eq.~\eqref{eq:affine} are called \emph{affine Boolean functions}.
For affine Boolean functions, both the globally optimal success probability (via the Helstrom bound) and the greedy success probability can be computed in closed form.
A direct calculation using properties of the trace norm shows that they are equal for all $s \in (0,1)$.
See \cref{subsec:greedy=opt} for details.

Our final result shows that this class is essentially the only one for which the greedy strategy is optimal, establishing a converse to \cref{thm:sufficient}.

\begin{restatable}[Necessity]{theorem}{necessity}\label{thm:necessary}
    Suppose a Boolean function $f : \bit^n \to \bit$ is \emph{not} affine, meaning it is \emph{not} of the form in \cref{eq:affine}.
    Then, the greedy strategy is not globally optimal, that is
    \begin{equation}
        \P(n,f,s,\greedy) < \P(n,f,s,\glo)
    \end{equation}
    for all but possibly finitely many values of $s \in (0,1)$.
\end{restatable}

Together, \cref{thm:sufficient,thm:necessary} provide a complete characterization: affine Boolean functions are precisely those for which memoryless measurement strategies suffice for optimal learning.
Note that \cref{thm:necessary} holds for all $s \in (0,1)$ except possibly for a finite set, hence for all but a set of measure zero.
The proof is presented in \cref{sec:necessary}.
This is followed by conclusions in \cref{sec:discussion}.

\Cref{app:boolean-grams,app:hypercube} contain supplementary results for the proof of \cref{thm:necessary}.
\Cref{app:non-affine-examples} briefly discusses the performance of the greedy strategy for some non-affine functions.
We consider the $\AND$ and $\MAJ$ (majority) functions as representative non-affine examples.
The former is highly unbalanced ($\abs{S_0}=2^n-1, \abs{S_1}=1)$, while the latter is balanced ($\abs{S_0}=\abs{S_1})$.
In both cases, globally optimal strategies outperform the greedy strategy, as illustrated by numerical results for small sequence lengths ($3\le n\le 9$).
Their asymptotic behavior, however, differs: for $\AND$, both the greedy and global strategies achieve perfect success as $n\to\infty$, whereas for $\MAJ$, the greedy strategy remains bounded away from unity in this limit.

\section{Proofs of Theorems~\ref{thm:PGM} and~\ref{thm:sufficient}: Performance of the greedy strategy}\label{sec:easy-proofs}
The \emph{pretty good measurement} (PGM), introduced by Hausladen and Wootters~\cite{hausladen1994pretty} and further analyzed by Barnum and Knill~\cite{barnum2002reversing}, provides a measurement strategy for quantum state discrimination.
Given an ensemble $\{(p_i, \rho_i)\}_{i=1}^n$, the PGM is defined by the measurement operators
\begin{equation}\label{eq:pgm-operators}
    M_i^{\text{PGM}} = \rho_{\text{avg}}^{-1/2} (p_i \rho_i) \rho_{\text{avg}}^{-1/2},
\end{equation}
where $\rho_{\text{avg}} = \sum_{i=1}^n p_i \rho_i$ and $\rho_{\text{avg}}^{-1/2}$ denotes its Moore-Penrose pseudo-inverse.
These operators form a valid POVM.
Barnum and Knill~\cite{barnum2002reversing} gave the following bound on the performance of the PGM.

\begin{theorem}[Barnum-Knill Inequality \cite{barnum2002reversing}]\label{thm:barnum-knill}
    Let $\P_{\text{opt}}$ denote the optimal success probability and $\P_{\text{PGM}}$ denote the success probability achieved by the pretty good measurement for discriminating an ensemble $\{(p_i, \rho_i)\}_{i=1}^n$.
    Then,
    \begin{equation}
        \P_{\text{opt}}^2 \leq \P_{\text{PGM}} \leq \P_{\text{opt}}.
    \end{equation}
\end{theorem}

This bound ensures that the PGM always achieves a success probability at least as large as the square of the optimal probability, guaranteeing near-optimal performance.
Moreover, in many cases of interest, the PGM is exactly optimal \cite{mochon2006family,dalla2015optimality,eldar2001quantum}.

\subsection{The greedy strategy}

Let $\mathcal{M} = \{M_0, M_1\}$ be an optimal POVM for discriminating the ensemble $\left\{ \left( 1/2, \ket{\psi_0} \right), \left( 1/2, \ket{\psi_1} \right) \right\}$ on a single qubit.
Since $\ket{\psi_0}$ and $\ket{\psi_1}$ appear with equal probabilities, the PGM is optimal for this single-qubit discrimination problem~\cite{hausladen1994pretty}, with
\begin{equation}
    M_i = \frac{1}{2} \rho^{-1/2} \ketbra{\psi_i} \rho^{-1/2},
\end{equation}
for $i \in \bit$ where $\rho = \frac{1}{2}\ketbra{\psi_0} + \frac{1}{2}\ketbra{\psi_1}$ is the average single-qubit state.

The greedy strategy is as follows: measure each qubit of the sequence $\ket{\psi_x}$ independently via the POVM $\mathcal{M}$, obtaining outcome $y_i \in \bit$ for the $i$-th qubit, yielding a classical string $y = y_1 \cdots y_n$.
Then, evaluate $f(y)$: if $y \in S_0$, guess that the input was from $S_0$; if $y \in S_1$, guess $S_1$.

The success probability of this strategy can be computed as follows.
Given that the true input is $x$, the probability of obtaining outcome string $y$ through local measurements is
\begin{equation}
    \P(y \mid x) = \Tr\left(M_y \ketbra{\psi_x}\right) = \prod_{i=1}^n \Tr(M_{y_i} \ketbra{\psi_{x_i}}),
\end{equation}
where $M_y = \bigotimes_{i=1}^n M_{y_i}$.
We succeed if $x$ and $y$ belong to the same set ($S_0$ or $S_1$).
Averaging over the uniform distribution on inputs, the greedy success probability is
\begin{equation}\label{eq:local_success}
    \P(n,f,s,\greedy)
    = \sum_{x \in \bit^n} \frac{1}{2^n} \sum_{\substack{y \in \bit^n \\ f(y)\mspace{2mu} = f(x)}} \P(y \mid x)
    = \frac{\abs{S_0}}{2^n} \sum_{y \in S_0} \Tr(M_y \sigma_0) + \frac{\abs{S_1}}{2^n} \sum_{y \in S_1} \Tr(M_y \sigma_1).
\end{equation}

\subsection{Lower-bounding the greedy strategy via the PGM}\label{subsec:greedy=pgm}

Now consider distinguishing the Boolean ensemble in \cref{eq:boolean-ensemble} via the PGM.
The average state is
\begin{equation}
    \rho_{\text{avg}}
    = p_0 \sigma_0 + p_1 \sigma_1 = \frac{\abs{S_0}}{2^n} \sigma_0 + \frac{\abs{S_1}}{2^n} \sigma_1
    = \frac{1}{2^n} \!\sum_{x \in \bit^n}\! \ketbra{\psi_x}
    = \Bigl(\frac{1}{2}\ketbra{\psi_0} + \frac{1}{2}\ketbra{\psi_1}\Bigr)^{\otimes n}\!\!
    = \rho^{\otimes n}.
\end{equation}
The PGM operators for distinguishing $\sigma_0$ and $\sigma_1$ are
\begin{equation}
    M_{\sigma_0}^{\text{PGM}}
    = \frac{\abs{S_0}}{2^n} \rho_{\text{avg}}^{-1/2} \sigma_0 \rho_{\text{avg}}^{-1/2}
    \quad \text{and} \quad
    M_{\sigma_1}^{\text{PGM}}
    = \frac{\abs{S_1}}{2^n} \rho_{\text{avg}}^{-1/2} \sigma_1 \rho_{\text{avg}}^{-1/2}.
\end{equation}
We can now establish the key connection:

\begin{theorem}\label{thm:pgm-local}
    For any Boolean function $f:\bit^n \to \bit$ and any inner product $s \in (0,1)$, the pretty good measurement achieves the same success probability as the greedy strategy for evaluating $f$.
\end{theorem}

\begin{proof}
    We show that the PGM operators can be decomposed as a sum of product measurements that precisely match the greedy strategy.
    For $i \in \bit$, starting with $M_{\sigma_i}^{\text{PGM}}$, we have
    \begin{subequations}
    \begin{align}
        M_{\sigma_i}^{\text{PGM}}
        &= \frac{\abs*{S_i}}{2^n} (\rho^{\otimes n})^{-1/2} \sigma_i (\rho^{\otimes n})^{-1/2}
        = \frac{\abs*{S_i}}{2^n} (\rho^{-1/2})^{\otimes n} \biggl(\frac{1}{\abs*{S_i}} \sum_{x \in S_i} \ketbra{\psi_x}\biggr) (\rho^{-1/2})^{\otimes n} \\
        &= \frac{1}{2^n} \sum_{x_1\ldots x_n \in S_i} \left(\rho^{-1/2} \ketbra{\psi_{x_1}} \rho^{-1/2}\right) \otimes \cdots \otimes \left(\rho^{-1/2} \ketbra{\psi_{x_n}} \rho^{-1/2}\right) \\
        &= \sum_{y_1 \ldots y_n \in S_i} M_{y_1} \otimes \cdots \otimes M_{y_n}
        = \sum_{y \in S_i} M_y.
    \end{align}
    \end{subequations}
    The success probability of the PGM is therefore
    \begin{subequations}\label{eq:local=pgm}
    \begin{align}
        \P(n,f,s,\mathrm{PGM}) &= p_0 \Tr(M_{\sigma_0}^{\text{PGM}} \sigma_0) + p_1 \Tr(M_{\sigma_1}^{\text{PGM}} \sigma_1) \\
        &= \frac{\abs{S_0}}{2^n} \sum_{y \in S_0} \Tr(M_y \sigma_0) + \frac{\abs{S_1}}{2^n} \sum_{y \in S_1} \Tr(M_y \sigma_1) \\
        &= \P(n,f,s,\greedy),
    \end{align}
    \end{subequations}
    where the last equality is by \cref{eq:local_success}.
\end{proof}
The proof of \cref{thm:PGM} follows immediately.
\begin{proof}[Proof of \cref{thm:PGM}]
    From \cref{thm:barnum-knill,thm:pgm-local} we have
    \begin{equation}
        \P(n,f,s,\glo)^2 \leq \P(n,f,s,\mathrm{PGM}) = \P(n,f,s,\greedy),
    \end{equation}
    completing the proof.
\end{proof}

\Cref{thm:pgm-local} reveals a remarkable structural property: although the PGM operators $M_{\sigma_i}^{\text{PGM}}$ are formulated as general operators on the full $n$-qubit Hilbert space, for the ensemble arising from Boolean function evaluation they decompose exactly into the local measurement operators of the greedy strategy:
\begin{equation}
    M_{\sigma_i}^{\text{PGM}} = \sum_{y_1 \ldots y_n \in S_i} M_{y_1} \otimes \cdots \otimes M_{y_n}.
\end{equation}
This decomposition shows that what appears to be a global measurement is, in fact, equivalent to independently measuring each qubit optimally.
Therefore, $\P(n,f,s,\greedy) = \P(n,f,s,\mathrm{PGM})$, and \cref{thm:PGM} follows immediately by applying the Barnum-Knill inequality (\cref{thm:barnum-knill}).

\subsection{Optimality for affine functions}\label{subsec:greedy=opt}

Having provided a universal guarantee for the greedy strategy for arbitrary Boolean functions, we find a class of Boolean functions for which it is optimal.
This is the content of \cref{thm:sufficient}, which we restate for convenience and prove below.
\sufficient*

\begin{proof}
    Any affine Boolean function is either constant or balanced.
    We consider the cases separately.

    \medskip
    \noindent\textbf{Case 1: Constant functions.}
    If $b_i = 0$ for all $i \in \{1,\ldots,n\}$, then $f(x) = b_0$ is constant.
    In this case, $S_0$ and $S_1$ are trivially distinguishable (one is empty and the other contains all strings), and both global and greedy strategies achieve success probability $1$.
    Thus, the greedy strategy is optimal.

    \medskip
    \noindent\textbf{Case 2: Balanced functions.}
    Without loss of generality\footnote{If $b_0 = 1$, then we are effectively just flipping the value of $f$ and also our guess.}, we assume $b_0=0$ and at least one $b_i \neq 0$ for $i \in \{1,\ldots,n\}$.
    Let $I = \{i \in [n] : b_i = 1\}$ denote the set of indices where the coefficient is $1$, and let $m = \abs{I} \geq 1$.
    Then,
    \begin{equation}
        f(x_1,\ldots,x_n) = \bigoplus_{i \in I} x_i.
    \end{equation}
    The set $S_0$ consists of all strings where $\bigoplus_{i \in I} x_i = 0$, and $S_1$ of all strings where $\bigoplus_{i \in I} x_i = 1$.
    Thus, $\abs{S_0} = \abs{S_1} = 2^{n-1}$.
    The optimal global success probability is given by the Helstrom measurement:
    \begin{equation}
        \P(n,f,s,\glo) = \frac{1}{2} + \frac{1}{4}\norm*{\sigma_0 - \sigma_1}_1.
    \end{equation}
    We compute
    \begin{subequations}
    \begin{align}
        \sigma_0 - \sigma_1 &= \frac{1}{2^{n-1}} \sum_{x \in S_0} \ketbra{\psi_x} - \frac{1}{2^{n-1}} \sum_{x \in S_1} \ketbra{\psi_x} \\
        &= \frac{1}{2^{n-1}} \sum_{x \in \bit^n} (-1)^{f(x)} \ketbra{\psi_x} \\
        &= \frac{1}{2^{n-1}} \sum_{x \in \bit^n} (-1)^{\bigoplus_{i \in I} x_i} \ketbra{\psi_x}.
    \end{align}
    \end{subequations}
    For the variables indexed by $i \in I$, we can factor out their contribution.
    For variables with $i \notin I$, summing over their values contributes a factor of $\left( \ketbra{\psi_0} + \ketbra{\psi_1} \right)$ for each such variable.
    Thus,
    \begin{subequations}
    \begin{align}
        \sigma_0 - \sigma_1 &= \frac{1}{2^{n-1}} \sum_{x \in \bit^n} (-1)^{\bigoplus_{i \in I} x_i} \ketbra{\psi_x} \\
        &= \frac{1}{2^{n-1}} \sum_{x_I \in \bit^m} \sum_{x_{\bar{I}} \in \bit^{n-m}} \Biggl[\bigotimes_{i \in I} (-1)^{x_i} \ketbra{\psi_{x_i}}\Biggr] \otimes \Biggl[\bigotimes_{j \notin I} \ketbra{\psi_{x_j}}\Biggr] \\
        &= \frac{1}{2^{n-1}} \Biggl[\bigotimes_{i \in I} \left( \ketbra{\psi_0} - \ketbra{\psi_1} \right) \Biggr] \otimes \Biggl[\bigotimes_{j \notin I} \left( \ketbra{\psi_0} + \ketbra{\psi_1} \right) \Biggr].
    \end{align}
    \end{subequations}
    The second equality uses the property $(-1)^{\bigoplus_{i \in I} x_i} = \bigotimes_{i \in I} (-1)^{x_i}$.
    Summation over $x_I$ and $x_{\bar{I}}$ indicates choosing $x_i$ where $i \in I$ and $j \notin I$, respectively.
    Using multiplicativity of the trace norm $\norm{A \otimes B}_1 = \norm{A}_1 \cdot \norm{B}_1$ and noting that $\norm[\Big]{ \ketbra{\psi_0} + \ketbra{\psi_1} }_1 = 2$ (since both are density operators),%
    \begin{subequations}
    \begin{align}
        \norm*{\sigma_0 - \sigma_1}_1
        &= \frac{1}{2^{n-1}} \norm[\Big]{ \ketbra{\psi_0} - \ketbra{\psi_1} }_1^m \cdot \norm[\Big]{ \ketbra{\psi_0} + \ketbra{\psi_1} }_1^{n-m} \\
        &= \frac{1}{2^{n-1}} \norm[\Big]{ \ketbra{\psi_0} - \ketbra{\psi_1} }_1^m \cdot 2^{n-m} \\
        &= \frac{1}{2^{m-1}} \norm[\Big]{ \ketbra{\psi_0} - \ketbra{\psi_1} }_1^m.
    \end{align}
    \end{subequations}
    Therefore,
    \begin{equation}
        \P(n,f,s,\glo) = \frac{1}{2} + \frac{1}{2^{m+1}} \norm[\Big]{ \ketbra{\psi_0} - \ketbra{\psi_1} }_1^m.
    \end{equation}

    Now, consider the greedy strategy.
    Let $p$ be the optimal success probability for discriminating $\left\{ \frac{1}{2}\ket{\psi_0}, \frac{1}{2}\ket{\psi_1} \right\}$ on a single qubit.
    The greedy strategy measures each qubit independently and computes $f$ on the resulting classical string $y$.
    Success occurs when $f(y) = f(x)$, which happens when $\bigoplus_{i \in I} y_i = \bigoplus_{i \in I} x_i$.
    For the $m$ qubits indexed by $I$, the parity of measurement outcomes must match the parity of inputs.
    For the remaining $n-m$ qubits, their outcomes are irrelevant.
    The probability that the parity is correct among the $m$ relevant qubits is
    \begin{subequations}
    \begin{align}
        \sum_{\substack{j=0 \\ j \text{ even}}}^m \binom{m}{j} p^{m-j}(1-p)^j
        &= \frac{1}{2} \left( \sum_{j=0}^m \binom{m}{j} p^{m-j}(1-p)^j + \sum_{j=0}^m (-1)^j \binom{m}{j} p^{m-j}(1-p)^j \right) \\[-.75em]
        &= \frac{1}{2} \left( (p + (1-p))^m + (p - (1-p))^m \right)
        = \frac{1}{2} + \frac{1}{2}(2p - 1)^m.
    \end{align}
    \end{subequations}
    For the optimal single-qubit measurement, $p = \frac{1}{2} + \frac{1}{4} \norm[\Big]{ \ketbra{\psi_0} - \ketbra{\psi_1} }_1$, resulting in \\ $2p - 1 = \frac{1}{2} \norm[\Big]{ \ketbra{\psi_0} - \ketbra{\psi_1} }_1$.
    Therefore,
    \begin{subequations}
    \begin{align}
        \P(n,f,s,\greedy) &= \frac{1}{2} + \frac{1}{2^{m+1}} \norm[\Big]{ \ketbra{\psi_0} - \ketbra{\psi_1} }_1^m \\
        &= \P(n,f,s,\glo),
    \end{align}
    \end{subequations} for all affine Boolean functions $f$.
\end{proof}

\section{Proof of Theorem~\ref{thm:necessary}: Only affine functions achieve optimality}\label{sec:necessary}
We now investigate the Boolean functions for which the PGM (and hence the greedy strategy) is globally optimal.

\subsection{Optimality criterion for the PGM}

The following lemma, due to Iten, Renes and Sutter~\cite{iten2016pretty}, provides a necessary and sufficient condition for the optimality of the PGM:

\begin{lemma}[PGM optimality~\cite{iten2016pretty}]\label{lem:pgmopt}
    Let $\{(\eta_x, \rho_x)\}$ be an ensemble of quantum states, and let $\ket{\xi_x}_{BB'}$ denote a purification of $(\rho_x)_B$ for each $x$.
    Then the pretty good measurement is optimal for distinguishing $\{(\eta_x, \rho_x)\}$ if and only if
    \begin{equation}\label{eq:PGM}
        [G_{X'B'},\tau_{X'B'}]=0,
    \end{equation}
    where $G_{X'B'}$ is the \emph{generalized Gram matrix} defined as
    \begin{equation}\label{eq:gengram}
        G_{X'B'} \coloneqq \sum_{x,x'} \sqrt{\eta_x \eta_{x'}} \ketbra{x}{x'}_{X'} \otimes \Tr_B\left(\ketbra{\xi_x}{\xi_{x'}}_{BB'}\right)
    \end{equation}
    and $\tau_{X'B'}$ is the block-diagonal operator
    \begin{equation}\label{eq:tau}
        \tau_{X'B'} \coloneqq \sum_{x'} \ketbra{x'}{x'}_{X'} \otimes \bra{x'}\sqrt{G_{X'B'}}\ket{x'}_{X'},
    \end{equation}
    formed from the diagonal blocks of $\sqrt{G_{X'B'}}$.
\end{lemma}

This criterion is powerful because it reduces the question of PGM optimality to checking whether two specific operators commute.
For our problem, the generalized Gram matrix of the Boolean ensemble is given by the following Lemma, the proof of which is in \cref{app:boolean-grams}.

\begin{restatable}{lemma}{booleangrams}\label{lem:gram-structure}
    The generalized Gram matrix $G$ for the Boolean ensemble has the block form
    \begin{equation}\label{eq:G-block}
        G = \frac{1}{2^n}
        \begin{pmatrix}
            \Gamma_0 & \Gamma' \\
            \Gamma'^{*} & \Gamma_1
        \end{pmatrix}
    \end{equation}
    where:
    \begin{itemize}
        \item $\Gamma_i$ is the Gram matrix of the states in $S_i$, with entries $(\Gamma_i)_{xy} = \braket{\psi_x}{\psi_y}$ for $x,y \in S_i$,
        \item $\Gamma'$ is the cross-Gram matrix with entries $(\Gamma')_{xy} = \braket{\psi_x}{\psi_y}$ for $x \in S_0$ and $y \in S_1$.
    \end{itemize}
\end{restatable}

\subsection{Structural constraints from PGM optimality}

We now translate the PGM optimality condition into an explicit algebraic relation between the Gram matrices $\Gamma_0$, $\Gamma_1$, and $\Gamma'$.

\begin{theorem}\label{thm:gamma-comm}
    If the PGM is optimal for evaluating a Boolean function $f$, then the Gram matrices satisfy
    \begin{equation}\label{eq:gamma-comm}
        \Gamma_0 \Gamma' = \Gamma' \Gamma_1.
    \end{equation}
\end{theorem}

\begin{proof}
    By \cref{lem:pgmopt}, the PGM is optimal if and only if $[G, \tau] = 0$, where $\tau$ is the block-diagonal operator formed from the diagonal blocks of $\sqrt{G}$.
    Ignoring the global factor $1/2^n$ (which does not affect the \gammacomm{}),
    \begin{align}
        G &=
        \begin{pmatrix}
            \Gamma_0 & \Gamma' \\
            \Gamma'^{*} & \Gamma_1
        \end{pmatrix},
        &
        \sqrt{G} &=
        \begin{pmatrix}
            X & Y \\
            Y^{*} & Z
        \end{pmatrix},
        &
        \tau &=
        \begin{pmatrix}
            X & 0 \\
            0 & Z
        \end{pmatrix}
    \end{align}
    where $X$ and $Z$ are positive semidefinite matrices, and $Y$ satisfies the constraint imposed by $\left( \sqrt{G} \right)^2 = G$.
    Expanding this constraint block-wise, we have
    \begin{align}\label{eq:block}
        X^2 + Y Y^{*} &= \Gamma_0, &
        X Y + Y Z &= \Gamma', &
        Z^2 + Y^{*} Y &= \Gamma_1.
    \end{align}
    Computing both sides of $G\tau = \tau G$, we have
    \begin{align}
        G\tau &=
        \begin{pmatrix}
            \Gamma_0 X & \Gamma' Z \\
            \Gamma'^{*} X & \Gamma_1 Z
        \end{pmatrix},
        &
        \tau G &=
        \begin{pmatrix}
            X \Gamma_0 & X \Gamma' \\
            Z \Gamma'^{*} & Z \Gamma_1
        \end{pmatrix}.
    \end{align}
    Equating the $(1,2)$-blocks yields
    \begin{equation}\label{eq:key}
        \Gamma' Z = X \Gamma'.
    \end{equation}
    We substitute the second equation of \cref{eq:block}, $\Gamma' = XY + YZ$, into \cref{eq:key} and simplify to get
    \begin{equation}
        YZ^2 = X^2 Y.
    \end{equation}
    Applying the other two equations of \cref{eq:block}, we have
    \begin{equation}
        Y(\Gamma_1 - Y^{*}Y) = (\Gamma_0 - YY^{*})Y,
    \end{equation}
    which we expand and cancel the common term $YY^{*}Y$ to obtain
    \begin{equation}\label{eq:YG}
        Y\Gamma_1 = \Gamma_0 Y.
    \end{equation}
    Finally, we compute $\Gamma_0\Gamma'$ using $\Gamma' = XY + YZ$ and the relations we have established.
    Note that $[\Gamma_0, X] = 0$ follows from comparing the $(1,1)$-blocks of $G\tau = \tau G$, and $[Z, \Gamma_1] = 0$ follows similarly from the $(2,2)$-blocks.
    Therefore,
    \begin{equation}
        \Gamma_0 \Gamma' = \Gamma_0 XY + \Gamma_0 YZ
        = X\Gamma_0 Y + Y\Gamma_1 Z = XY\Gamma_1 + YZ\Gamma_1 = \Gamma' \Gamma_1,
    \end{equation}
    completing the proof.
\end{proof}

\subsection{Consequences of the matrix relation}

Having established that PGM optimality implies $\Gamma_0 \Gamma' = \Gamma' \Gamma_1$, we now extract the structural consequences of this relation.
A key observation is that the inner product between any two product states can be expressed in terms of the inner product $s = \braket{\psi_0}{\psi_1}$ and the Hamming distance:
\begin{equation}\label{eq:inner-product-hamming}
    \braket{\psi_x}{\psi_y} = \prod_{i=1}^n \braket{\psi_{x_i}}{\psi_{y_i}} = s^{d(x,y)},
\end{equation}
where $d(x,y)$ is the Hamming distance between two $n$-bit strings $x$ and $y$.
This means that every entry of the matrices $\Gamma_0$, $\Gamma_1$, and $\Gamma'$ is a monomial in $s$, allowing us to view these matrices as matrix-valued polynomials in $s$.
The relation $\Gamma_0\Gamma' = \Gamma'\Gamma_1$ thus becomes a polynomial identity that must hold for all $s \in (0,1)$.

The polynomial structure of the matrix relation $\Gamma_0\Gamma' = \Gamma'\Gamma_1$ yields two immediate and powerful consequences.
First, by comparing the number of monic terms in corresponding matrix entries, we show that the function $f$ must be constant ($\abs{S_0}\abs{S_1}=0$) or balanced ($\abs{S_0} = \abs{S_1}$).

\begin{theorem}\label{thm:balanced}
    If the \gammacomm{} in \cref{eq:gamma-comm} holds, then the function $f$ must be either constant $(\abs{S_0}\abs{S_1} = 0)$ or balanced $(\abs{S_0} = \abs{S_1})$, except possibly for a finite set of values of $s \in (0,1)$.
\end{theorem}
\begin{proof}
    By \cref{eq:inner-product-hamming}, every entry of the matrices $\Gamma_0$, $\Gamma_1$, and $\Gamma'$ is a monic monomial in $s = \braket{\psi_0}{\psi_1}$.
    Thus, every term in $ \Gamma_0 \Gamma' $ is a sum of $ \abs{S_0} $ monic monomials, and every term in $ \Gamma' \Gamma_1 $ is a sum of $ \abs{S_1} $ monic monomials.
    We aim to show that $\Gamma_0 \Gamma' = \Gamma' \Gamma_1$ implies the function to be constant or balanced, which we do via its contrapositive.
    If the function is not constant, i.e., $\abs{S_0}\abs{S_1} \neq 0$, then $\Gamma_0 \Gamma'$ and $\Gamma' \Gamma_1$ are non-empty.
    If the function is additionally not balanced, i.e., $\abs{S_0} \neq \abs{S_1}$, then the elements of $\Gamma_0 \Gamma'$ and $\Gamma' \Gamma_1$ cannot coincide, except possibly for finitely many values of $s$.
\end{proof}
Two comments are in order.
First, the implication of this theorem does not hold for some $s$ due to the \gammacomm{} being true regardless of the function structure.
Additionally, the finite set of $s$ exceptions depends on $f$.
In this and the following theorems, this clause means:
    Let $f$ be a Boolean function.
    For all except possibly finitely many $s \in (0,1)$, the implication holds.
Second, this constant or balanced consequence is used in \cref{thm:counting} below, and later in \cref{lem:hypercube-flip}.
As their statements are trivial for constant functions, we assume balanced functions when the \gammacomm{} holds.
We do this for ease of comprehension, noting that the same reasoning may hold for arbitrary functions when accounting for vacuously true statements.

The second consequence could be understood in terms of the \emph{Boolean hypercube} graph (see \cref{app:hypercube} for details), and we use it shortly after to derive details about function structure.
In this graph, nodes are $n$-bit strings and edges connect two nodes if and only if they differ by exactly one coordinate.
For any pair of vertices $x \in S_0$ and $y \in S_1$ in the Boolean hypercube, consider all intermediate vertices $w$ such that a path from $x$ to $y$ via $w$ has total length $L$ (where the path need not be shortest overall, just the sum $d(x,w) + d(w,y) = L$).
By equating coefficients of powers of $s$ in the polynomial identity, we show that, for balanced functions, the number of such intermediate vertices in $S_0$ equals the number in $S_1$, for every choice of $x$, $y$ and $L$.

\begin{theorem}\label{thm:counting}
    If the \gammacomm{} in \cref{eq:gamma-comm} holds and the function $f$ is balanced, then for all integers $L \in \{0,1,\ldots,2n\}$ and for all $x \in S_0$, $y \in S_1$,
    \begin{equation}\label{eq:counting-equality}
        \abs{\{\,u \in S_0 : d(x,u) + d(u,y) = L\,\}}
        =
        \abs{\{\,v \in S_1 : d(x,v) + d(v,y) = L\,\}}
    \end{equation}
    for all but possibly finitely many values of $s \in (0,1)$.
\end{theorem}

\begin{proof}
    Fix any arbitrary $x \in S_0$ and $y \in S_1$.
    The $(x,y)$-entries of the matrix products in \cref{eq:gamma-comm} are
    \begin{align}
        (\Gamma_0 \Gamma')_{xy} &= \sum_{u \in S_0} s^{d(x,u) + d(u,y)}&
        (\Gamma' \Gamma_1)_{xy} &= \sum_{v \in S_1} s^{d(x,v) + d(v,y)}.
    \end{align}
    Equality of these expressions as polynomials in $s$ means
    \begin{equation}
        \sum_{u \in S_0} s^{d(x,u) + d(u,y)} = \sum_{v \in S_1} s^{d(x,v) + d(v,y)}.
    \end{equation}
    We can rewrite each sum by grouping terms according to the total distance $L = d(x,w) + d(w,y)$,%
    \begin{equation}\label{eq:counting-polynomial}
        \sum_{L=0}^{2n} s^L \abs{\{u \in S_0 : d(x,u) + d(u,y) = L\}}
        =
        \sum_{L=0}^{2n} s^L \abs{\{v \in S_1 : d(x,v) + d(v,y) = L\}}.
    \end{equation}
    This is a polynomial identity in $s$.
    For this to hold, either the coefficients of each power $s^L$ on both sides must agree (giving \cref{eq:counting-equality}), or $s$ must be one of (at most $2n$) roots of the above polynomial.
\end{proof}

\subsection{Necessity of the affine structure}

We now combine the results from the previous sections to obtain a complete characterization of when the greedy strategy is optimal.
A crucial element in our proof is the following lemma about the Boolean hypercube, the proof of which can be found in \cref{app:hypercube}.

\begin{restatable}[Hypercube flip characterization]{lemma}{hypercubeflip}\label{lem:hypercube-flip}
    Let $Q_n$ be the $n$-dimensional hypercube graph, and let $A, B \subseteq V(Q_n)$ be a partition satisfying
    \begin{enumerate}[label=(\roman*)]
        \item $A \cup B = V(Q_n)$ and $A \cap B = \varnothing$;
        \item $\abs{A} = \abs{B} = 2^{n-1}$ (balanced partition);
        \item\label{item:counting} for all $L \in \N$ and for all $x \in A$, $y \in B$,
        \begin{equation}\label{eq:count-condition}
            \abs{\{\,u \in A: d(x,u) + d(u,y) = L\,\}}
            =
            \abs{\{\,v \in B: d(x,v) + d(v,y) = L\,\}}.
        \end{equation}
    \end{enumerate}
    Then there exists a unique index $i \in [n]$ such that for every $x = (x_1,\ldots,x_n) \in \bit^n$,
    \begin{equation}
        x \in A \quad\Longleftrightarrow\quad x^{(i)} \in B,
    \end{equation}
    where $x^{(i)} = (x_1,\ldots,x_{i-1},\overline{x_i},x_{i+1},\ldots,x_n)$ denotes the bit-flip of $x$ in the $i$-th coordinate.
\end{restatable}
We start with the following theorem, which states that the functions for which the greedy strategy is optimal must necessarily follow a recursive structure.
\begin{theorem}\label{thm:function-form}
    If the \gammacomm{} in \cref{eq:gamma-comm} holds, then the function $f$ is either constant or there exists an index $ i \in [n] $ and a Boolean function
    $ g : \bit^{n-1} \to \bit $ such that
    \begin{equation}\label{eq:function-xor}
        f(x_1,x_2,\ldots,x_n)
        =
        g(x_1,\ldots,x_{i-1},x_{i+1},\ldots,x_n) \oplus x_i
    \end{equation}
    for all but possibly finitely many values of $s \in (0,1)$.
\end{theorem}
\begin{proof}
    Starting from the \gammacomm{}, \cref{thm:balanced} gives that the function is either constant (in which case the proof ends) or balanced.
    From \cref{thm:counting}, the \gammacomm{} additionally means that for all
    $x \in S_0$, $y \in S_1$ and integers $L$,%
    \begin{equation}\label{eq:counting-balance}
        \abs{\{\,u \in S_0 : d(x,u) + d(u,y) = L\,\}}
        =
        \abs{\{\,v \in S_1 : d(x,v) + d(v,y) = L\,\}}
    \end{equation}
    for all but possibly finitely many values of the inner product $s$.
    By the hypercube structure of $\bit^n$ and \cref{lem:hypercube-flip}, the above two consequences imply that there exists an index $ i \in [n] $ such that flipping the $i$-th bit maps $S_0$ bijectively onto $S_1$.
    Equivalently, for every $x=(x_1,\ldots,x_n)\in\bit^n$,
    \begin{equation}
        x \in S_0
        \quad\Longleftrightarrow\quad
        x^{(i)} \in S_1,
    \end{equation}
    where $x^{(i)}$ denotes the bit-flip of $x$ in the $i$-th coordinate.
    Define
    \begin{equation}
        g(x_1,\ldots,x_{i-1},x_{i+1},\ldots,x_n)
        \coloneqq f(x_1,\ldots,x_{i-1},0,x_{i+1},\ldots,x_n).
    \end{equation}
    Then, by construction,
    \begin{equation}
        f(x_1,\ldots,x_n)
        =
        g(x_1,\ldots,x_{i-1},x_{i+1},\ldots,x_n) \oplus x_i,
    \end{equation}
    which proves \cref{eq:function-xor}, as needed.
\end{proof}

The recursive structure revealed in \cref{thm:function-form} allows us to apply induction to completely characterize optimal functions by proving \cref{thm:necessary}.

We briefly outline the essential steps of the proof for the reader's convenience.
We have previously seen an equivalent condition for PGM optimality in \cref{lem:pgmopt}, from which we derived a necessary condition in \cref{thm:gamma-comm}: the \gammacomm{} in \cref{eq:gamma-comm}.
The main consequence of this is \cref{thm:function-form}, which constrains the function to a recursive form given by \cref{eq:function-xor}, which when shown via induction for each $n$ implies the function is affine.
In particular, our inductive hypothesis is: if $g$ is an $(n-1)$-bit Boolean function for which \cref{eq:gamma-comm} holds, then $g$ must be affine.
The remaining part of the induction step is to show the antecedent of the above, which we do in \cref{lem:reduced-optimality}.
Specifically, if \cref{eq:gamma-comm,eq:function-xor} hold for an $n$-bit Boolean function $f$, then the $(n-1)$-bit Boolean function $g$ that appears in the recursive form of $f$ must satisfy the \gammacomm{} in \cref{eq:gamma-comm}.
Therefore any Boolean function satisfying \cref{eq:gamma-comm} must be affine, implying that $f$ itself must be affine.

\necessity*

\begin{proof}
    We prove this statement via its contrapositive.
    Since the greedy strategy is globally optimal for evaluating a Boolean function $f$, \cref{thm:gamma-comm} implies that the associated Gram matrices satisfy $\Gamma_0 \Gamma' = \Gamma' \Gamma_1$.
    Therefore, it suffices to show that this in turn means the function $f$ must be an affine Boolean function.
    We prove this statement by induction on the number of variables $n$.

    \medskip
    \noindent\textbf{Base case}: The case when $n=1$ is trivial.
    For illustrative purposes, we additionally show the $n=2$ case.
    By \cref{thm:balanced}, $f$ is either constant or balanced.
    If $f$ is constant, the claim is again trivial.
    If $f$ is balanced then $\abs{f^{-1}(0)} = 2$.
    Without loss of generality, assume that $f(0,0) = 0$; otherwise, we may replace $f$ by its complement $\overline{f}$, which is equivalent to flipping the value of $b_0$.
    The possible choices for $f^{-1}(0) \subseteq \{00, 01, 10, 11\}$ are:
    $\{00, 01\}, \{00, 10\},$ and $\{00, 11\}$.
    Respectively, these correspond to the functions
    \begin{align}
        f(x_1, x_2) &= x_1,&
        f(x_1, x_2) &= x_2,&
        f(x_1, x_2) &= x_1 \oplus x_2.
    \end{align}
    All three are affine, confirming the claim for $n = 2$.

    \medskip
    \noindent\textbf{Induction step}:
    Let $n \geq 2$ be an integer and assume the inductive hypothesis for $n-1$, which is that if the Gram matrices of an $(n-1)$-bit Boolean function satisfy the \gammacomm{} in \cref{eq:gamma-comm}, then that function is affine.

    Consider an arbitrary $n$-bit Boolean function $f$ that (by the antecedent of the statement to be proven) satisfies the \gammacomm{} in \cref{eq:gamma-comm}.
    By \cref{thm:function-form}, this function $f$ is either constant (and thus affine with $b_1=\ldots=b_n=0$ and $b_0$ equal to the function value, meaning we are done), or there exists an index $i \in [n]$ and an $(n-1)$-bit Boolean function $g:\bit^{n-1} \to \bit$ such that
    \begin{equation}
        f(x_1,\ldots,x_n) = g(x_1,\ldots,x_{i-1},x_{i+1},\ldots,x_n) \oplus x_i.
    \end{equation}
    These two facts let us show the antecedent of the induction hypothesis via the following statement.

    \begin{lemma}\label{lem:reduced-optimality}
        Let $G_0$ (resp.~$G_1$) be the Gram matrix of the states corresponding to inputs of $g$ that map to $0$ (resp.~$1$), and let $G'$ be the cross-Gram matrix.
        If $f$ has the form in \cref{eq:function-xor} and satisfies the \gammacomm{} in \cref{eq:gamma-comm}, then $g$ also satisfies the \gammacomm{}, i.e., $G_0 G' = G' G_1$.
    \end{lemma}

    \begin{proof}
        We have that $\Gamma_0 \Gamma' = \Gamma' \Gamma_1$.
        Order the elements of $S_0$ by first listing those with $x_i=0$ and then those with $x_i=1$.
        Analogously, order $S_1$ with $x_i=1$ elements first and $x_i=0$ second.
        With respect to this ordering, the Gram matrices have the block form
        \begin{align}
            \Gamma_0 &= \Gamma_1 =
            \begin{pmatrix}
                G_0 & s G' \\
                s G'^{*} & G_1
            \end{pmatrix},&
            \Gamma' &=
            \begin{pmatrix}
                s G_0 & G' \\
                G'^{*} & s G_1
            \end{pmatrix},
        \end{align}
        where $s=\braket{\psi_0}{\psi_1}$.
        We now compute both sides of \cref{eq:gamma-comm}.
        First,
        \begin{equation}
            \Gamma_0 \Gamma'
            =
            \begin{pmatrix}
                s G_0^2 + s G' G'^{*}
                &
                G_0 G' + s^2 G' G_1
                \\[4pt]
                s^2 G'^{*} G_0 + G_1 G'^{*}
                &
                s G'^{*} G' + s G_1^2
            \end{pmatrix}.
        \end{equation}
        Similarly,
        \begin{equation}
            \Gamma' \Gamma_1
            =
            \begin{pmatrix}
                s G_0^2 + s G' G'^{*}
                &
                s^2 G_0 G' + G' G_1
                \\[4pt]
                G'^{*} G_0 + s^2 G_1 G'^{*}
                &
                s G'^{*} G' + s G_1^2
            \end{pmatrix}.
        \end{equation}
        Equality of these two matrices implies equality of their $(1,2)$-blocks, hence
        \begin{equation}
            G_0 G' + s^2 G' G_1
            =
            s^2 G_0 G' + G' G_1 .
        \end{equation}
        Rearranging gives
        \begin{equation}
            (1-s^2)\, G_0 G' = (1-s^2)\, G' G_1 .
        \end{equation}
        Since $s\in(0,1)$, we conclude
        \begin{equation}
            G_0 G' = G' G_1 ,
        \end{equation}
        as claimed.
    \end{proof}

    Returning to the main proof: since we previously showed that $f$ satisfies \cref{eq:gamma-comm,eq:function-xor}, the just-proven \cref{lem:reduced-optimality} implies that $g$ also satisfies \cref{eq:gamma-comm}.
    By the induction hypothesis, $g$ is then affine, meaning it has the form
    \begin{equation}
        g(x_1,\ldots,x_{i-1},x_{i+1},\ldots,x_n) = b_0 \oplus \bigoplus_{j \neq i} b_j x_j.
    \end{equation}
    Substituting into \cref{eq:function-xor},
    \begin{equation}
        f(x_1,\ldots,x_n) = \Bigl( b_0 \oplus \bigoplus_{j \neq i} b_j x_j \Bigr) \oplus x_i = b_0 \oplus \bigoplus_{j=1}^n b_j x_j,
    \end{equation}
    where we set $b_i = 1$ and $b_j$ for $j \neq i$ as in $g$.
    This is precisely the affine form in \cref{eq:affine}, completing the induction and the proof.
\end{proof}

\section{Conclusion}\label{sec:discussion}
We characterized when memoryless measurement strategies suffice for learning Boolean properties of quantum sequences.
Given a sequence of $n$ qubits, each in one of two known states $\ket{\psi_0}$ or $\ket{\psi_1}$ having inner product $s \in (0,1)$, the problem is to determine the value of a Boolean function of the (unknown) underlying bit string.
We posed this as a state discrimination problem between two mixed states corresponding to the two values of the Boolean function, and studied the minimum-error regime.

We considered a simple greedy strategy which measures each qubit independently to optimally distinguish $\ket{\psi_0}$ from $\ket{\psi_1}$ requiring no quantum or classical memory.
We showed that for evaluating affine Boolean functions, this strategy is equivalent to the optimal global strategy which could utilize quantum memory and joint measurements.
We also established the converse of this, i.e., any Boolean function that is not affine cannot be learned optimally using this greedy strategy (except for possibly at finitely many values of $s$).
This suggests that all other Boolean properties fundamentally require memory for optimal learning.

Beyond this characterization, we proved that the greedy strategy always achieves at least the square of the optimal success probability, establishing it as the pretty good measurement for this discrimination problem.
This guarantees that even when not optimal, memoryless strategies remain competitive.

Previous works have explored cryptographic primitives in the bounded quantum-storage model \cite{damgaard2008cryptography} and state discrimination with limited quantum memory \cite{ballester2008state}.
Since both quantum memory and joint measurements remain technologically challenging \cite{lvovsky2009optical,ma2021one,conlon2023approaching}, characterizing when memoryless strategies suffice is of significant interest.
Our results could inform settings such as quantum reading of classical memories \cite{pirandola2011quantum}, and quantum-enhanced pattern recognition \cite{ortolano2023quantum}.

There are several avenues for future research.
The first is to study the greedy strategy for learning Boolean properties of quantum sequences for other priors (e.g., unequal, correlated) or other settings, such as unambiguous discrimination.
The second is to introduce limited classical memory to allow some degree of adaptation.
It is of interest to investigate which other classes of Boolean functions can be optimally learned using adaptive local strategies.
Finally, studying online or streaming strategies where qubits arrive sequentially and the Boolean function value is continuously updated as the sequence grows, is another direction for future research.

\subsection*{Software}
Companion software for computing optimal discrimination probabilities, including the semidefinite programs used in this work, is available in the \texttt{toqito} quantum information package~\cite{russo2021toqito}.

\subsection*{Acknowledgements}
We thank Gary Au for helpful discussions on some of our combinatorial arguments.
J.S.~and A.Z.~are funded in part by the Commonwealth of Virginia's Commonwealth Cyber Initiative (CCI) under grant number 469351.
T.G.~acknowledges the Indian Statistical Institute, Kolkata for support during the early stages of this work.
T.G.~acknowledges financial support from the ANRF National Post-Doctoral Fellowship (NPDF) under File No.~PDF/2025/005147.

\bibliographystyle{unsrt-links}
\bibliography{bib-minerr}

\appendix
\crefalias{section}{appendix}

\section{The Gram matrices of the Boolean ensemble}\label{app:boolean-grams}
To apply the PGM optimality criterion from \cref{lem:pgmopt}, we must compute the generalized Gram matrix $G$ for the Boolean ensemble (\cref{eq:boolean-ensemble}).
We begin by constructing explicit purifications of the mixed states $\sigma_0$ and $\sigma_1$.

For each $i \in \bit$, a purification of $\sigma_i$ is given by
\begin{equation}
    \ket{\xi_i}_{BB'} = \frac{1}{\sqrt{\abs*{S_i}}} \sum_{x \in S_i} \ket{\psi_x}_B \ket{x}_{B'},
\end{equation}
where the auxiliary system $B'$ has dimension $\abs*{S_i}$, and the states $\{\ket*{y}_{B'}\}_{y \in S_i}$ form an orthonormal basis of $B'$.
Here, $B$ denotes the physical $n$-qubit system containing the state $\ket{\psi_x}$, while $B'$ is the purifying system that tracks which element of $S_i$ was selected.

Recall from \cref{eq:gengram} that the generalized Gram matrix is defined as
\begin{equation}
    G_{X'B'} \coloneqq \sum_{i,j \in \bit} \sqrt{p_i p_j} \ketbraij \otimes \Tr_B\left(\ketbra{\xi_i}{\xi_j}_{BB'}\right),
\end{equation}
where $i, j \in \bit$ index the two hypotheses (whether the input belongs to $S_0$ or $S_1$), and $p_x$ are the prior probabilities: $p_0 = \abs{S_0}/2^n$ and $p_1 = \abs{S_1}/2^n$.

We compute the partial trace:
\begin{equation}
    \Tr_B\left(\ketbra{\xi_i}{\xi_j}_{BB'}\right)
    = \Tr_B\left(\frac{1}{\sqrt{\abs*{S_i}\abs*{S_j}}} \sum_{\substack{x \in S_i \\ y \in S_j}} \ketbra{\psi_x}{\psi_y}_B \otimes \ketbra*{x}{y}_{B'}\right)
    = \frac{1}{\sqrt{\abs*{S_i}\abs*{S_j}}} \sum_{\substack{x \in S_i \\ y \in S_j}} \braket{\psi_y}{\psi_x} \!\ketbra*{x}{y}_{B'}.
\end{equation}

Substituting into the definition of $G$ and using $\sqrt{p_i p_j} = \sqrt{\abs*{S_i}\abs*{S_j}}/2^n$, we obtain
\begin{subequations}
\begin{align}
    G_{X'B'} &= \sum_{i,j \in \bit} \sqrt{p_i p_j} \ketbraij \otimes \Tr_B\left(\ketbra{\xi_i}{\xi_j}_{BB'}\right) \\
    &= \sum_{i,j \in \bit} \frac{\sqrt{\abs*{S_i}\abs*{S_j}}}{2^n} \ketbraij \otimes \frac{1}{\sqrt{\abs*{S_i}\abs*{S_j}}} \sum_{\substack{x \in S_i \\ y \in S_j}} \braket{\psi_y}{\psi_x} \!\ketbra*{x}{y}_{B'} \\
    &= \frac{1}{2^n} \sum_{i,j \in \bit} \!\ketbraij \otimes \sum_{\substack{x \in S_i \\ y \in S_j}} \braket{\psi_y}{\psi_x} \ketbra*{x}{y}_{B'}.
\end{align}
\end{subequations}
This calculation reveals a natural block structure for $G$, which allows us to prove \cref{lem:gram-structure}.

\booleangrams*

\begin{proof}
    The block structure follows directly from the calculation above.
    The $(i,j)$-block of $G$ corresponds to the operator acting on $\ketbraij$, which contains the matrix of inner products $\braket{\psi_y}{\psi_x}$ for $x \in S_i$ and $y \in S_j$.
    Since the states have real inner products, we use the identity $\braket{\psi_y}{\psi_x} = \braket{\psi_x}{\psi_y}$ to match the order in the standard definition of Gram matrices.
    For $i = j$, this gives the Gram matrices $\Gamma_0$ and $\Gamma_1$.
    For $i \neq j$, this gives the cross-Gram matrices $\Gamma'$ and $\Gamma'^*$.
    The global factor $1/2^n$ appears uniformly across all blocks.
\end{proof}

\section{The Boolean hypercube graph}\label{app:hypercube}
The structure of Boolean functions is intimately connected to the geometry of the $n$-dimensional hypercube.
For $n \in \N$, the \emph{$n$-dimensional hypercube graph} $Q_n$ is defined as follows:
\begin{itemize}
    \item The vertex set is $V(Q_n) = \bit^n$, the set of all $n$-bit strings.
    \item Two vertices $x = (x_1,\ldots,x_n)$ and $y = (y_1,\ldots,y_n)$ are adjacent if and only if they differ in exactly one coordinate.
\end{itemize}

The \emph{Hamming distance} $d(x,y)$ between two vertices $x$ and $y$ is the number of coordinates in which they differ, which equals the graph distance in $Q_n$:
\begin{equation}
    d(x,y) = \abs{\{i \in [n] : x_i \neq y_i\}},
\end{equation}
where $[n] = \{1, \ldots, n\}$.
We denote adjacency by `$\sim$', meaning $x \sim y$ if and only if $d(x, y) = 1$.
The complete graph on two vertices is denoted as $K_2$, i.e., the graph with vertex set
$\bit$ with a single edge connecting them.
A perfect matching in a graph is a set of pairwise disjoint edges such that every vertex of the graph is incident to exactly one edge in the set.

An important property of the hypercube is its inductive structure via the Cartesian product.
Recall that the Cartesian product $G \cartprod H$ of two graphs has vertex set $V(G) \times V(H)$, where $(g,h) \sim (g',h')$ if either $g = g'$ and $h \sim h'$ in $H$, or $h = h'$ and $g \sim g'$ in $G$.
The hypercube satisfies the recursive relation
\begin{equation}\label{eq:cart-prod}
    Q_{n+1} \cong Q_n \cartprod K_2,
\end{equation}
meaning $Q_{n+1}$ can be viewed as two copies of $Q_n$ connected by a perfect matching.
More precisely, if we partition $V(Q_{n+1})$ according to the last bit, we obtain two sub-cubes $Q_n \times \{0\}$ and $Q_n \times \{1\}$, each isomorphic to $Q_n$, with matching edges connecting vertices that differ only in the last coordinate.
\Cref{fig:hypercube} illustrates this structure for $Q_3$.

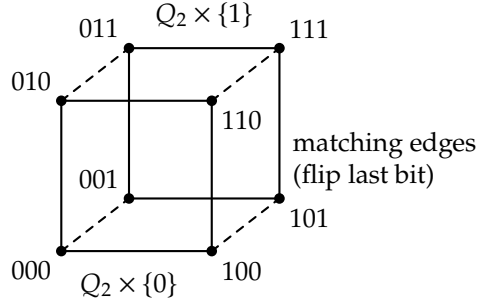
\begin{figure}[t]
    \centering
    \begin{tikzpicture}[
        scale=1,
        line cap=round,
        line join=round,
        every node/.style={font=\small}
    ]
        \coordinate (A0) at (0,0);
        \coordinate (B0) at (2,0);
        \coordinate (C0) at (2,2);
        \coordinate (D0) at (0,2);
        \coordinate (A1) at (0.9,0.7);
        \coordinate (B1) at (2.9,0.7);
        \coordinate (C1) at (2.9,2.7);
        \coordinate (D1) at (0.9,2.7);
        \draw[thick] (A0)--(B0)--(C0)--(D0)--cycle;
        \draw[thick] (A1)--(B1)--(C1)--(D1)--cycle;
        \draw[dashed, thick] (A0)--(A1);
        \draw[dashed, thick] (B0)--(B1);
        \draw[dashed, thick] (C0)--(C1);
        \draw[dashed, thick] (D0)--(D1);
        \foreach \p in {A0,B0,C0,D0,A1,B1,C1,D1}{
            \fill (\p) circle (2pt);
        }
        \node[below left]  at (A0) {$000$};
        \node[below right] at (B0) {$100$};
        \node[below right] at (C0) {$110$};
        \node[above left]  at (D0) {$010$};
        \node[above left]  at (A1) {$001$};
        \node[below right] at (B1) {$101$};
        \node[above right] at (C1) {$111$};
        \node[above left]  at (D1) {$011$};
        \node at (0.9,-0.45) {$Q_2 \times \{0\}$};
        \node at (1.9,3.15) {$Q_2 \times \{1\}$};
        \node[align=left] at (4.3,1.2) {matching edges\\(flip last bit)};
    \end{tikzpicture}
    \caption{%
        The three-dimensional hypercube $Q_3 \cong Q_2 \cartprod K_2$, visualized as two copies of $Q_2$ with a perfect matching (dashed edges) connecting corresponding vertices that differ in the last bit coordinate.
    }
    \label{fig:hypercube}
\end{figure}
\vspace{0.2cm}
\noindent
The following lemma characterizes balanced partitions of the hypercube satisfying the path-counting symmetry \cref{eq:count-condition} which are precisely those obtained by flipping a single coordinate.%
\hypercubeflip*

To prove \cref{lem:hypercube-flip}, we need the following lemma which establishes that the two parts must have adjacent vertices.

\begin{lemma}\label{lem:min-dist}
    Let $A, B \subseteq V(Q_n)$ satisfy the hypotheses of \cref{lem:hypercube-flip}.
    Then
    \begin{equation}
        \min\{d(x, y) : x \in A, \, y \in B\} = 1.
    \end{equation}
\end{lemma}

\begin{proof}
    Let $k = \min\{d(x, y) : x \in A, \, y \in B\}$.
    Assume for contradiction that $k > 1$.
    Choose $x \in A$ and $y \in B$ with $d(x, y) = k$, and let
    $x$--$x_1$--$\,\cdots$--$x_{k-1}$--$y$ be a shortest path from $x$ to $y$ in $Q_n$.

    Then, $d(x, x_1) = 1$ and $d(x_1, y) = k - 1 < k$.
    Since $A$ and $B$ form a partition, for any $x_1 \in V(Q_n)$ we either have:
    \begin{itemize}
        \item $x_1 \in A$, meaning $d(x_1, y) < k$ with $x_1 \in A$ and $y \in B$, contradicting minimality of $k$, or
        \item $x_1 \in B$, meaning $d(x, x_1) = 1 < k$ with $x \in A$ and $x_1 \in B$, again contradicting minimality.
    \end{itemize}
    Hence, $k = 1$.
\end{proof}

\smallskip

\begin{proof}[Proof of \cref{lem:hypercube-flip}]
    The case when $n = 1$ is trivial.
    For illustrative purposes, we additionally show the case when $n = 2$.

    The vertex set of the $2$-dimensional hypercube is $V(Q_2) = \{00, 01, 10, 11\}$.
    By \cref{lem:min-dist}, there exist vertices $x \in A$ and $y \in B$ such that $d(x, y) = 1$.
    Without loss of generality, choose $x = 00$ and $y = 01$.
    Since $\abs{A} = \abs{B} = 2$, the only possible balanced partitions with $00 \in A$ and $01 \in B$ are
    \begin{equation}
        A = \{00, 10\}, \; B = \{01, 11\} \qquad \text{or} \qquad A = \{00, 11\}, \; B = \{01, 10\}.
    \end{equation}
    In both cases, $B$ is obtained from $A$ by flipping the second bit.
    Thus, the claim holds for $n = 2$.

    \medskip

    We next prove the statement for an arbitrary positive integer $n$.
    By \cref{lem:min-dist}, choose vertices $x \in A$ and $y \in B$ such that $d(x, y) = 1$.
    Let $i \in [n]$ be the unique coordinate where $x$ and $y$ differ, so $y = x^{(i)}$.
    We now show that coordinate $i$ implements the flip globally.

    \medskip
    \noindent Take any coordinate $k \neq i$.
    The four vertices $x, \; x^{(i)}, \; x^{(i,k)} = (x^{(i)})^{(k)}, \; x^{(k)}$ form a 4-cycle in $Q_n$:
    \begin{equation}
        x \;-\; x^{(i)} \;-\; x^{(i,k)} \;-\; x^{(k)} \;-\; x.
    \end{equation}
    We now consider two cases based on whether $x^{(i,k)}$ belongs to $A$ or $B$.

    \medskip
    \noindent \textit{Case 1: $x^{(i,k)} \in B$.} Apply \cref{eq:count-condition} at $L = 2$ to the antipodal pair $\bigl(x, \, x^{(i,k)}\bigr)$.
    The only two length-2 paths between them pass through $x^{(i)}$ and $x^{(k)}$ respectively.
    Since $x^{(i,k)} \in B$, balancing the count forces $x^{(k)} \in A$.
    Thus, we have
    \begin{equation}
        x \in A, \quad x^{(i)} \in B, \quad x^{(i,k)} \in B, \quad x^{(k)} \in A.
    \end{equation}
    The pair $\bigl(x^{(i,k)}, x^{(k)}\bigr)$ differs only in coordinate $i$, just as $(x, x^{(i)})$ does.

    \medskip
    \noindent \textit{Case 2: $x^{(i,k)} \in A$.}
    If $x^{(k)} \in A$, then it belongs to a different partition than $x^{(i)}$, and we may apply \cref{eq:count-condition} with $L = 2$ to the antipodal pair $\bigl(x^{(i)}, x^{(k)}\bigr)$.
    The only two length-2 paths between them pass through $x$ and $x^{(i,k)}$, both in $A$, giving two $A$-intermediates and zero $B$-intermediates, which leads to a contradiction.
    Hence, we must have $x^{(k)} \in B$, meaning
    \begin{equation}
        x \in A, \quad x^{(i)} \in B, \quad x^{(i,k)} \in A, \quad x^{(k)} \in B.
    \end{equation}
    Again, the pair $\bigl(x^{(i,k)}, x^{(k)}\bigr)$ differs only in coordinate $i$, just as $(x, x^{(i)})$ does.

    \medskip
    \noindent We have established the following {propagation property}: if $x \in A$ and $x^{(i)} \in B$, then for every coordinate $k \neq i$, the vertices $x^{(k)}$ and $x^{(i,k)}$ belong to distinct partitions.
    In other words, flipping any non-$i$ coordinate sends a split pair to a split pair, preserving the same flip coordinate $i$.

    Since every vertex pair $z, z^{(i)} \in V(Q_n)$ is reachable from the pair $x, y$ by a finite sequence of non-$i$ coordinate flips, we conclude that
    \begin{equation}
        \forall x \in V(Q_n), \qquad x \in A \iff x^{(i)} \in B,
    \end{equation}
    ending the proof.
\end{proof}

\section{Greedy strategy performance for non-affine functions}\label{app:non-affine-examples}
In this section, we examine two important non-affine functions: the $\AND$ function and the majority function.
For these functions, the greedy strategy is suboptimal, but understanding the gap between greedy and global performance provides insight into when adaptive or joint measurements offer practical advantages.

\subsection{The \texorpdfstring{$\boldsymbol{\AND}$}{AND} function}

The $\AND$ function on $n$ bits is defined by
\begin{equation}
    \AND(x_1,\ldots,x_n) = x_1 \wedge x_2 \wedge \cdots \wedge x_n = \begin{cases}
        1 & \text{if } x_i = 1 \text{ for all } i, \\
        0 & \text{otherwise}.
    \end{cases}
\end{equation}
This function is highly unbalanced: $S_1 = \{11\ldots1\}$ contains only the all-ones string, while $S_0$ contains all other $2^n - 1$ strings.

\subsubsection{Greedy success probability}

For the greedy strategy, we measure each qubit independently using the optimal single-qubit measurement, which succeeds with probability $p = \frac{1}{2} + \frac{1}{4} \norm[\Big]{ \ketbra{\psi_0} - \ketbra{\psi_1} }_1 = \frac{1}{2}(1+\sqrt{1-s^2})$ in distinguishing $\ket{\psi_0}$ from $\ket{\psi_1}$ with equal priors.
The greedy strategy fails in evaluating $\AND$ when
\begin{itemize}
    \item $x = 11\ldots1$ but we guess $y \neq 11\ldots1$, which happens with probability $(1-p^n)/2^n$.
    \item $x \neq 11\ldots1$ but we guess $y = 11\ldots1$, which happens with probability
\end{itemize}
\begin{equation}
    \sum_{x \neq 11\ldots1} \frac{1}{2^n} p^{w(x)} (1-p)^{n-w(x)}
    = \frac{1}{2^n} \sum_{k=0}^{n-1} \binom{n}{k} p^k (1-p)^{n-k}
    = \frac{1}{2^n} (1 - p^n),
\end{equation}
where $w(x)$ is the Hamming weight of $x$.
The greedy success probability is therefore
\begin{equation}
    \P(n,\AND,s,\greedy) = 1-\frac{1-p^n}{2^{n-1}}.
\end{equation}

\begin{figure}[t]
    \centering
    \includegraphics[width=0.8\linewidth]{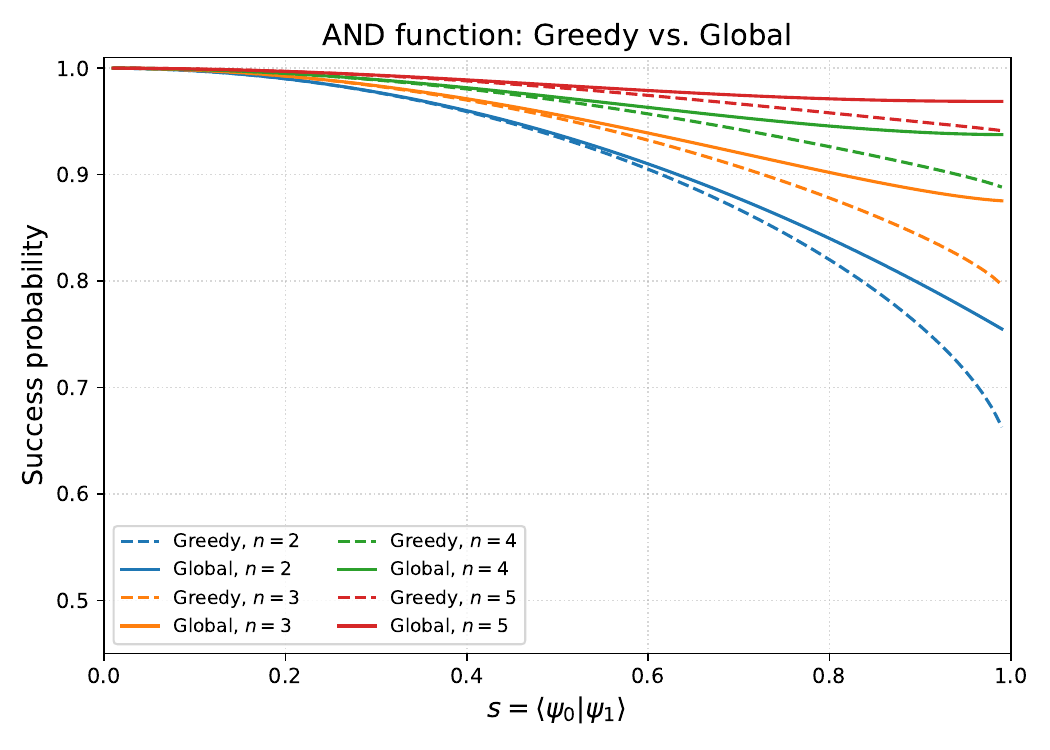}
    \caption{%
        Greedy vs.~global success probabilities for the $\AND$ function.
        As $n$ increases, both strategies approach perfect success exponentially fast, with the gap diminishing rapidly.}
    \label{fig:and-comparison}
\end{figure}

\subsubsection{Asymptotic behavior}

As $n \to \infty$ with $p$ fixed (corresponding to fixed non-orthogonal states $\ket{\psi_0}$ and $\ket{\psi_1}$), we have $\P(n,\AND,s,\greedy) \to 1$ since $1/2^{n-1} \to 0$ exponentially fast.

This limiting behavior is intuitive: with overwhelming probability $(2^n-1)/2^n \to 1$, the input contains at least one zero.
The probability of accidentally guessing all ones when the input is not all ones vanishes exponentially.
Thus, the greedy strategy succeeds with probability approaching $1$.
The global optimal strategy achieves slightly better performance by leveraging correlations across qubits, but this advantage becomes negligible as $n$ increases.
For large $n$, the extreme imbalance of the $\AND$ function dominates, and both strategies approach perfect success.

\subsection{The majority function}

For odd $n = 2k+1$, the majority function is defined by
\begin{equation}
    \MAJ(x_1,\ldots,x_{2k+1}) = \begin{cases}
        1 & \text{if } \sum_{i=1}^{2k+1} x_i > k, \\
        0 & \text{if } \sum_{i=1}^{2k+1} x_i \leq k.
    \end{cases}
\end{equation}
Unlike the $\AND$ function, majority is balanced: $\abs{S_0} = \abs{S_1} = 2^{2k}$.

\subsubsection{Greedy success probability}

We now compute the greedy success probability for the majority function.
The greedy strategy measures each qubit independently with the optimal single-qubit measurement, which correctly identifies the bit value with probability $p$ and incorrectly identifies it with probability $1-p$.

Given an input string $x$ with Hamming weight $w(x) = j$ (i.e., $j$ ones and $2k+1-j$ zeros), the measurement produces an output string $y$ whose Hamming weight $w(y)$ is a random variable.
We can decompose $w(y)$ as
\begin{equation}
    w(y) = Y_1 + Y_2,
\end{equation}
where
\begin{itemize}
    \item $Y_1$ is the number of ones in $y$ that came from ones in $x$ (correctly identified ones), with $Y_1 \sim \Binom(j, p)$;
    \item $Y_2$ is the number of ones in $y$ that came from zeros in $x$ (incorrectly flipped zeros), with $Y_2 \sim \Binom(2k+1-j, 1-p)$.
\end{itemize}

Since different qubits are measured independently, $Y_1$ and $Y_2$ are independent random variables.

The greedy strategy succeeds when $\MAJ(x) = \MAJ(y)$, which occurs when:
\begin{itemize}
    \item If $j \leq k$ (so $f(x) = 0$): we need $w(y) \leq k$, i.e., $Y_1 + Y_2 \leq k$.
    \item If $j > k$ (so $f(x) = 1$): we need $w(y) > k$, i.e., $Y_1 + Y_2 > k$.
\end{itemize}

Let $\gamma_j$ denote the conditional success probability given that $w(x) = j$.
For $j \leq k$:
\begin{equation}
    \gamma_j
    = \P(Y_1 + Y_2 \leq k)
    = \sum_{a=0}^{j} \sum_{b=0}^{\min(2k+1-j, k-a)} \binom{j}{a} p^a (1-p)^{j-a} \binom{2k+1-j}{b} (1-p)^b p^{2k+1-j-b}.
\end{equation}

For $j > k$ (equivalently, $j \geq k+1$):
\begin{equation}
    \gamma_j
    = \P(Y_1 + Y_2 > k)
    = \sum_{a=0}^{j} \sum_{b=\max(0,k+1-a)}^{2k+1-j} \binom{j}{a} p^a (1-p)^{j-a} \binom{2k+1-j}{b} (1-p)^b p^{2k+1-j-b}.
\end{equation}

By symmetry of the majority function, we have $\gamma_j = \gamma_{2k+1-j}$ for all $j$.
To see this, note that if we flip all bits in both $x$ and $y$, the majority value also flips, and the distribution of measurement outcomes remains the same (with the roles of 0 and 1 interchanged).
Therefore, the conditional success probability for $j$ ones is the same as for $j$ zeros (i.e., $2k+1-j$ ones).

The greedy success probability is:
\begin{equation}
    \P(n=2k+1,\MAJ,s,\greedy)
    = \sum_{j=0}^{2k+1} \frac{1}{2^{2k+1}} \binom{2k+1}{j} \gamma_j
    = \sum_{j=0}^{k} \frac{1}{2^{2k}} \binom{2k+1}{j} \gamma_j,
\end{equation}
where we used the symmetry $\gamma_j = \gamma_{2k+1-j}$ to combine terms.

\begin{figure}[t]
    \centering
    \includegraphics[width=0.8\linewidth]{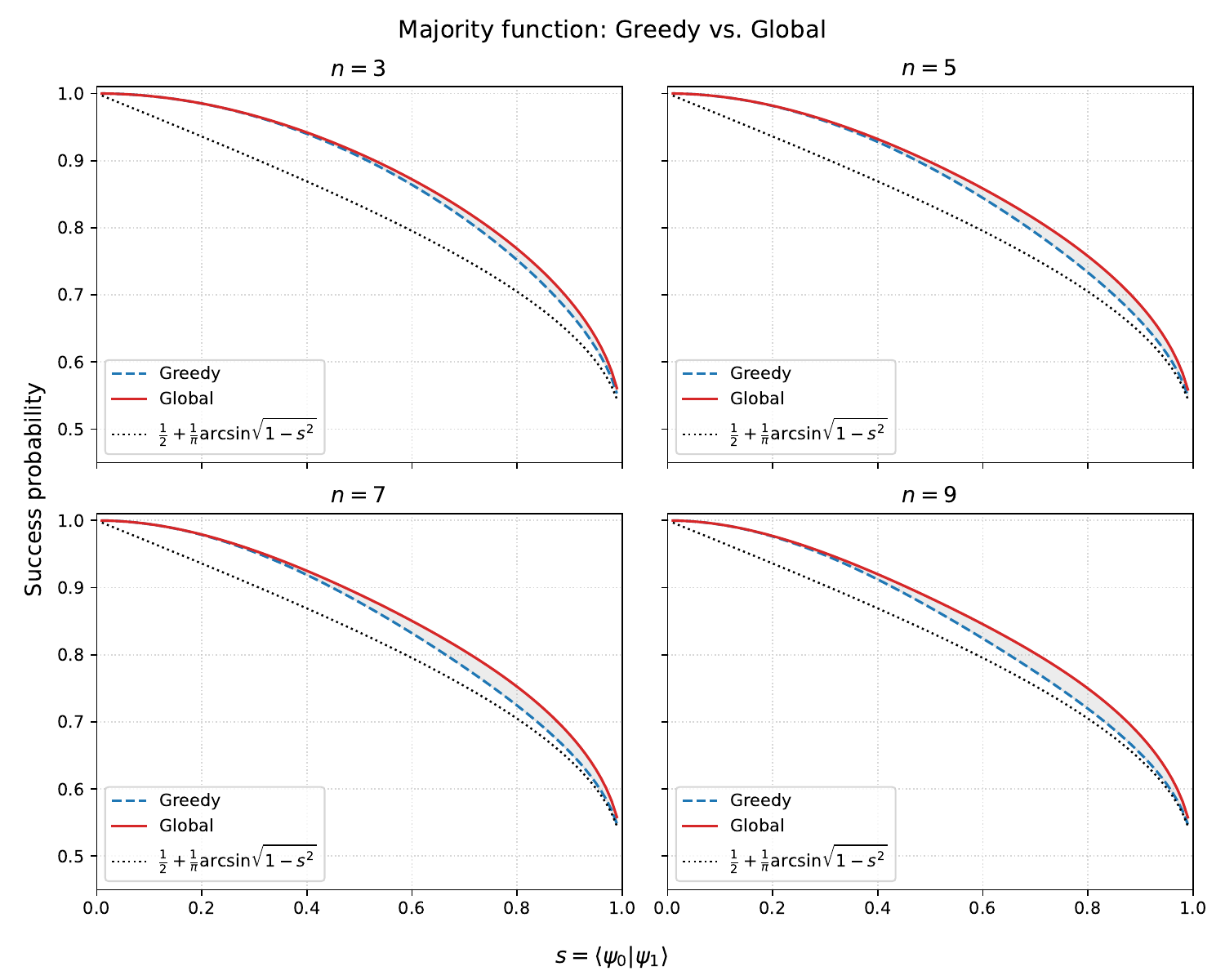}
    \caption{%
        Greedy vs.~global success probabilities for the majority function.
    }
    \label{fig:maj-comparison}
\end{figure}

\subsubsection{Asymptotic behavior}

The asymptotic behavior of the majority function is governed by its noise sensitivity.
Unlike the $\AND$ function, the failure probability for majority does not vanish as $n \to \infty$.

To analyze this, we relate our setting to the noise sensitivity framework.
When we measure each qubit with the optimal single-qubit measurement, we correctly identify each bit with probability $p=\frac{1}{2}(1+\sqrt{1-s^2})$ and incorrectly with probability $\epsilon \coloneqq 1-p$.
This is equivalent to applying independent random noise to each bit with noise rate $\epsilon$.

The noise sensitivity of the majority function has been precisely characterized \cite{mossel2003noise}.
For the majority function on $n = 2k+1$ bits with noise parameter $\epsilon$, the probability that the function value changes (i.e., the failure probability) is given by
\begin{equation}
    \P_{\text{fail}}^\MAJ = \frac{1}{2}- \frac{1}{\pi}\arcsin({1-2\epsilon}) + O\left(\frac{1}{\sqrt{n\epsilon}}\right).
\end{equation}
In the limit as $n \to \infty$ with small $\epsilon$, this approaches
\begin{equation}
    \P_{\text{fail}}^\MAJ \to \frac{2}{\pi}\sqrt{\epsilon} = \frac{2}{\pi}\sqrt{1-p}.
\end{equation}
Therefore, the local success probability approaches
\begin{equation}
    \P(n,\MAJ,s,\greedy) \to\frac{1}{2}+ \frac{1}{\pi}\arcsin({\sqrt{1-s^2}})
\end{equation}
as $n \to \infty$.
This limiting value is strictly less than $1$ for any $p < 1$ (i.e., for non-orthogonal states).

\subsection{Comparison and discussion}

The contrasting behaviors of the $\AND$ and majority functions illustrate two distinct regimes:

\paragraph{Highly unbalanced functions (\texorpdfstring{$\boldsymbol{\AND}$}{AND}):} When $\abs{S_0} \gg \abs{S_1}$ or vice versa, the greedy strategy performs nearly optimally for large $n$ because the overwhelming size of one pre-image dominates the success probability.
Both greedy and global strategies converge to perfect success exponentially fast as $n \to \infty$, and the advantage from joint measurements becomes negligible in this regime.

\paragraph{Balanced non-affine functions (majority):} For balanced functions that are not affine, the situation is fundamentally different.
The majority function's noise sensitivity ensures that the greedy strategy converges to a limiting success probability strictly bounded away from $1$.
Specifically, $\P(n,\MAJ,s,\greedy) \to 1 - \frac{2}{\pi}\arcsin(\sqrt{1-p})$ as $n \to \infty$, which remains bounded below $1$ for any non-orthogonal states ($p < 1$), while the global optimal strategy achieves better performance.
This gap between greedy and global strategies persists even asymptotically, demonstrating a persistent advantage of joint measurements that does not vanish as the problem size increases.

These examples confirm that our characterization theorem captures a fundamental dichotomy: affine functions are the unique class where local memoryless measurements are always optimal, while all other functions require more complex strategies.
Crucially, the practical significance of this advantage is intimately tied to the function's noise sensitivity.
Low noise sensitivity functions (like $\AND$ for large $n$) exhibit vanishing advantage from joint measurements, while high noise sensitivity functions (like majority) maintain advantage from joint-measurements even for arbitrarily large input dimensions.
This connection between noise sensitivity and joint-measurement-advantage suggests deeper structural relationships between classical Boolean function complexity and quantum information-theoretic resources.

\end{document}